\documentclass[journal,twoside,web
]{ieeecolor}
\usepackage{generic}
\usepackage{cite}
\usepackage{amsmath,amssymb,amsfonts}
\usepackage{algorithmic}
\usepackage[dvipdfmx]{graphicx}
\usepackage{textcomp}

\def\BibTeX{{\rm B\kern-.05em{\sc i\kern-.025em b}\kern-.08em
    T\kern-.1667em\lower.7ex\hbox{E}\kern-.125emX}}
\newcommand\req[1]{(\ref{eq:#1})}
\newcommand\refig[1]{Fig.~\ref{fig:#1}}
\newcommand\reapp[1]{Appendix~\ref{app:#1}}

\newcommand\reas[1]{Assumption~\ref{as:#1}}
\newcommand\relem[1]{Lemma~\ref{lem:#1}}

\newcommand\reth[1]{Theorem~\ref{th:#1}}

\newcommand\ressec[1]{\ref{ssec:#1}}

\newtheorem{theorem}{Theorem}
\newtheorem{lemma}{Lemma}

\newtheorem{fact}{Fact}
\newtheorem{assumption}{Assumption}

\newtheorem{remark}{Remark}

\newcommand\bR{\mathbb{R}}

\newcommand\cA{\mathcal{A}}

\newcommand\cL{\mathcal{L}}

\newcommand\cS{\mathcal{S}}

\newcommand\bone{\textbf{1}}
\newcommand\veps{\varepsilon}
\def\nn{\nonumber}

\begin{document}
\title{
Passivity-based Analysis and Design for Population Dynamics with Conformity Biases
}
\author{
Shunya Yamashita, \IEEEmembership{Student Member, IEEE}, 
Kodai Irifune, 
Takeshi Hatanaka, \IEEEmembership{Senior Member, IEEE}, \\
Yasuaki Wasa, \IEEEmembership{Member, IEEE}, 
Kenji Hirata, \IEEEmembership{Member, IEEE} 
and 
Kenko Uchida, \IEEEmembership{Member, IEEE}
\thanks{
This work was supported in part by JST-Mirai Program under Grant 18077648, and in part by JSPS KAKENHI under Grant 21J13956.
}
\thanks{
S.~Yamashita, K.~Irifune and T.~Hatanaka are with the School of Engineering, 
Tokyo Institute of Technology, 
Tokyo 152-8550, Japan 
(e-mail: yamashita.s.ag@hfg.sc.e.titech.ac.jp; irifune@hfg.sc.e.titech.ac.jp; hatanaka@sc.e.titech.ac.jp). 
}
\thanks{
Y.~Wasa and K.~Uchida are with the School of Advanced Science and Engineering, 
Waseda University, 
Tokyo 169-8555, Japan 
(e-mail: wasa@aoni.waseda.jp; kuchida@waseda.jp).
}
\thanks{
K.~Hirata is with the Faculty of Engineering, University of Toyama, 
Toyama 930-8555, Japan 
(e-mail: hirata@eng.u-toyama.ac.jp).
}
}

\maketitle

\begin{abstract}
	This paper addresses mechanisms for boundedly rational decision makers in discrete choice problem. 
	First, we introduce two mathematical models of population dynamics with conformity biases. 
	We next analyze the models in terms of $\delta$-passivity, and show that the conformity biases work to break passivity of decision makers. 
	Based on the passivity perspective, we propose mechanisms so as to induce decision makers to a desired population state. 
	Furthermore, we analyze a convergence property of designed mechanisms, and present parameter conditions to guarantee stable inducements. 
\end{abstract}


\section{Introduction}
\label{sec:intro}



Human behavior is a critically important factor in design and analysis of large-scale social systems like transportation network~\cite{CSet2013} and energy management systems~\cite{DHP2002,BMet2010,LCD2015}. 
As seen in these publications, a typical approach to address such systems involving humans is to model social decisions/dynamics assuming human rationality, against the background of expected utility theory~\cite{S1991}. 
Meanwhile, behavioral economics has pointed out that human rationality may be bounded due to information and/or cognitive constraints, 
as typified by prospect theory~\cite{KT1979} or dual-process theory~\cite{BC2014}. 
It is highly uncertain that any system designed under the  assumption of human rationality will work when their rationality is bounded. 
Motivated by the issue, boundedly rational human/social models have begun to be investigated in the field of systems and control\cite{GAT2019,MSB2020}. 
Among various types of bounded rationality, in this paper, we focus on so-called conformity bias~\cite{KN2002}, tendency to follow the majority, that is observed in the scenes of parking location choice~\cite{FM2007} and evacuation decision~\cite{UH2012}.


One of the decision making issues is discrete choice problem which decision makers choose a strategy from finite number of options~\cite{T2003}. 
Indeed, some significant issues in society are categorized as discrete choice: route selection in transportation networks \cite{SM2000}, 
choice of energies \cite{S2017}, 
and water distribution \cite{RQ2010}, for example. 
Meanwhile, when we deal with systems including large scale of population, 
it should be required to consider decision models as a population \cite{MFLB2018}. 
Discrete choice behavior by large populations is well-addressed in evolutionary game theory, in which various types of population dynamics have been presented, e.g., logit dynamics~\cite{HH2005,HS2007}, Smith dynamics~\cite{S1984} and pairwise
comparison dynamics~\cite{S2010}. 
Meanwhile, 
some recent publications have pointed out relations between population dynamics and passivity \cite{FS2013,PSM2018,PPD2017}. 
The researches in \cite{FS2013,PSM2018} analyze passivity for generally formulated population dynamics, 
and \cite{PPD2017} particularly focuses on the decision model in water distribution system.
However, the models dealt with in the above literature do not explicitly consider influences of biases. 
Although several publications \cite{CL2018,ZYDZ2021} deal with bounded rationality in population dynamics, they address the models without focusing on passivity. 


In social systems, it is sometimes preferred to induce humans to a desired social state. 
From the viewpoint of inducement for populations, 
behavior modification mechanism might be a key concept to achieve a desired behavior. 
This includes incentive \cite{LM2002} which is an economic approach, or nudge \cite{TS2008} which is an informational one. 
These kinds of mechanisms have recently attracted the attention of control community \cite{RDO2014,SCM2020}, 
and more specifically, the ones for population dynamics have also been studied \cite{WEet2016,CL2016}. 
Whereas, the publication which explicitly consider bounded rationality is rare. 
The authors in \cite{CL2018} propose a nudging mechanism for biased population dynamics, and analyze stability by using singular perturbation.


In this paper, we address mechanisms for discrete choice problem under bounded rationality. 
First, we introduce a well-used population dynamics model, called the logit dynamics, and its $\delta$-passivity \cite{PSM2018}. 
Next, we extend the decision making model to the one with conformity biases. 
We show two types of the bias models respectively addressed in~\cite{BD2001} and~\cite{CL2018}. 
We then analyze the biased logit dynamics in terms of $\delta$-passivity, and reveal the impacts of conformity biases. 
Inspired by passivity-based control methods, we propose mechanisms so as to achieve desired social behavior for two types of biased models, respectively. 
Furthermore, we analyze their convergence, and show parameter conditions in order to stabilize the proposed mechanisms. 

We summarize the contributions of this paper as follows: 
\begin{itemize}
	\item the impact by conformity biases to population dynamics is shown in terms of passivity, and 
	\item passivity-based mechanisms are presented, which stably achieve desired social state. 
\end{itemize}
Parts of the contents in this paper are similar to the previous work in~\cite{YHet2020}. 
Meanwhile, variation of bias models and exact analysis of mechanisms are the incremental contributions added anew in this paper.

\section{Model Description}

\subsection{Preliminaries: $\delta$-Passivity}

This subsection introduces $\delta$-passivity \cite{PSM2018} which is a similar concept to passive systems \cite{HCFS2015}. 
Suppose a system $\Sigma$ represented by the state space model 
\begin{align*}
	\dot{x} = f(x, u), 
\end{align*}
where $x \in \bR^n$ is the state, $u \in \bR^n$ is the input, and $f : \bR^n \times \bR^n \to \bR^n$ is a function. 
Then, $\Sigma$ is called $\delta$-passive from $\dot{u}$ to $\dot{x}$ if there exists a positive semi-definite function $S : \bR^n \times \bR^n \to [0, \infty)$ and a scalar $\rho \geq 0$ such that
\begin{align*}
S(u(t), x(t)) - S(u(0), x(0)) \leq \int_{0}^{t} \dot{x}^\top (\hat{t}) \dot{u} (\hat{t}) - \rho ||\dot{u} (\hat{t})||^2 d \hat{t}
\end{align*}
for all input $u$, all initial state $x(0)$ and all $t \geq 0$. 
The positive semi-definite function $S$ is particularly called a storage function. 
In addition, $\Sigma$ is called $\delta$-input-strictly-passive if the above inequality is satisfied with some positive scalar $\rho > 0$. 
As widely known, if $S$ is differentiable, we can replace the above inequality with $\dot{S} (u(t), x (t)) \leq \dot{x}^\top (t) \dot{u}(t) - \rho ||\dot{u} (t)||^2$. 

In the same way as passivity shortage \cite{QS2014}, we define $\delta$-output-passivity-shortage. 
The system $\Sigma$ is called $\delta$-output-passivity-short from $\dot{u}$ to $\dot{x}$ if there exists a positive semi-definite function $S : \bR^n \times \bR^n \to  [0, \infty)$, and a scalar $\gamma \geq 0$ such that
\begin{align*}
S(u(t), x(t)) - S(u(0), x(0)) \leq \int_{0}^{t} \dot{x}^\top (\hat{t}) \dot{u} (\hat{t}) + \gamma ||\dot{x} (\hat{t})||^2 d \hat{t}
\end{align*}
for all input $u$, all initial state $x(0)$ and all $t \geq 0$. 
We call the value $\gamma$ as an impact coefficient. 
If $S$ is differentiable, the above inequality is equivalent to 
$\dot{S} (u(t), x (t)) \leq \dot{x}^\top (t) \dot{u}(t) + \gamma ||\dot{x} (t)||^2$.

Feedback system composed of strictly passive component and passivity-short one is related to the Lyapunov stability~\cite{HCFS2015,SJK2012}. 
Consider a $\delta$-input-strictly-passive system whose input is $u_1$ and state is $x_1$, which satisfies $\dot{S}_1 \leq \dot{x}_1^\top \dot{u}_1 - \rho \| \dot{u}_1 \|^2$ for a differentiable storage function $S_1$ and $\rho > 0$. 
We also suppose a $\delta$-output-passivity-short system whose input is $u_2$ and state is $x_2$, which satisfies $\dot{S}_2 \leq \dot{x}_2^\top \dot{u}_2 + \gamma \| \dot{x}_2 \|^2$ for a differentiable storage function $S_2$ and $\gamma > 0$. 
Then, the feedback interconnection under $u_1 = x_2$ and $u_2 = -x_1$ provides
\begin{align*}
	\dot{S}_1 + \dot{S}_2 \leq (\gamma - \rho) \| \dot{x}_1 \|^2. 
\end{align*}
If $\gamma \leq \rho$ holds and $S_1 + S_2$ is radially unbounded, the above inequality suggests stability of the feedback system in the sense of Lyapunov~\cite{HCFS2015}. 
In other words, $\delta$-input-strictly-passive systems can stabilize $\delta$-output-passivity-short ones by negative feedback.

\subsection{Dynamic Decision Making Model under Rationality}

This subsection introduces the logit dynamics~\cite{HS2007,PSM2018} which is a well-used dynamical discrete choice model.

Consider the situation that decision makers choose a strategy from $n$ available strategies. 
We denote the strategy set as $\cA := \{ 1, 2, \ldots, n \}$. 
Suppose that the population of the decision makers can be represented as a continuum value. 
We now define $\cS := \{ x \in [0,1]^n \mid \bone_n^\top x =1\}$ as the strategy choice distribution set, where $\bone_n$ describes $n$ dimensional vector whose elements are all~$1$. 
We denote the relative interior of the set $\cS$ as $\textrm{int}(\cS)$. 
The population state $\pi \in \cS$ implies the distribution of strategy choice. 
Specifically, the $k$-th element of $\pi$, denoted by $\pi_k \in [0, 1]$, implies a fraction of the decision makers selecting strategy $k \in \cA$. 

Let us define a cost vector $\tau \in \bR^n$ of which $k$-th element $\tau_k \in \bR$ corresponds to the cost for choosing strategy $k \in \cA$. 
Then, the logit dynamics~\cite{HS2007,PSM2018} is represented as 
\begin{align}
\dot{\pi}= \eta \left(Q(\tau)-\pi \right) , \ \pi(0) \in \textrm{int}(\cS), 
\label{eq:logit}
\end{align}
where $\eta$ is a positive constant which indicates the update rate. 
The function $Q: \bR^n \to \textrm{int}(\cS)$ corresponds to the steady state of \req{logit} and follows
\begin{align}
Q_k(\tau)=\frac{\exp(-\beta \tau_k)}{\sum\limits_{l \in \cA}\exp(-\beta \tau_l)}, 
\label{eq:logit_s}
\end{align}
where $Q_k : \bR^n \to (0, 1)$ implies $k$-th element of $Q$, and $\beta > 0$ is a constant. 
The system \req{logit} is known to guarantee $\pi(t) \in \textrm{int}(\cS)$ for any time $t \geq 0$. 
Remarking that $Q_k(\tau) > Q_i(\tau)$ for any pair $k \neq i$ satisfying $\tau_k < \tau_i$, 
the logit dynamics \req{logit} can be interpreted as an approximated model of best response, i.e., rational decision making \cite{HH2005}. 

The logit dynamics \req{logit} is also known to satisfy $\delta$-passivity as below. 
\begin{lemma}[\cite{PSM2018}]
\label{lem:dp0}
The logit dynamics \req{logit} is $\delta$-passive from $- \dot{\tau}$ to $\dot{\pi}$ for the storage function
\begin{align}
S(\tau,\pi) := \ & \eta \left( \pi^\top \tau+\frac{1}{\beta} \sum_{l \in \cA}\pi_l \log{\pi_l} \right)
\nonumber\\
& \ - \eta \min_{\omega \in \cS} \left( \omega^\top \tau+\frac{1}{\beta}\sum_{l \in \cA}\omega_l \log{\omega_l} \right), 
\label{eq:stfunc1}
\end{align}
i.e., $\dot{S} \leq - \dot{\tau}^\top \dot{\pi}$ holds. 
\end{lemma}
Remark that the storage function $S(\tau,\pi)$ has a kind of radial unboundedness. 
See \reapp{pr_Srad} for the details.

\subsection{Bias Models}

The logit model \req{logit} corresponds to an approximation of rational strategy choice. 
Meanwhile, in behavioral economics field, it has been pointed out that human's decision process is sometimes biased relying on population state. 
This subsection introduces two models of conformity bias. 
In particular, we represent these bias models as the ones to affect the cost $\tau$. 
Hereafter, we refer to $\tau$ as a biased cost.
Alternatively, the actual cost is denoted by $T \in \bR^n$, which is basically supposed to be non-negative. 

First, we introduce the bias model which is a generalization of interactions model~\cite{BD2001}. 
In this model, the biased cost $\tau$ is given by
\begin{align}
\tau = T + B (\pi), 
\label{eq:bias1}
\end{align}
where $B : \cS \to \bR^n$ is the bias function, whose $k$-th element obeys $b_k(\pi_k)$. 
The function $b_k : [0, 1] \to \bR$ is assumed to satisfy the following: 
\begin{assumption}
\label{as:func_b}
The function $b_k \ (k \in \cA)$ obeys the following items: 
\begin{itemize}
\item[(i)] $b_k$ is continuous in $[0, 1]$ and continuously differentiable in $(0, 1)$. 
\item[(ii)] 
$b_k$ is decreasing in $\pi_k$, and uniformly bounded by $[b^L, b^H]$ with some $b^H \geq b^L$. 
\item[(iii)] 
The first derivative $\nabla b_k$ is uniformly bounded by $[- c^H, - c^L]$ with some $c^H \geq c^L > 0$. 
\end{itemize}
\end{assumption}

The item (ii) in \reas{func_b} implies that if $k \in \cA$ is a majority strategy, the decision makers get a lower impression of cost than the actual cost $T_k$. 
Hence, $B(\pi)$ has the tendency that the strategy chosen by many people intensifies its own popularity, which corresponds to conformity bias. 
The first derivative $\nabla b_k$ implies the dependency of $b_k$ on $\pi_k$, i.e., this corresponds to the strength of the bias. 
In particular, $c^H$, which is the maximum of $| \nabla b_k (\pi_k) |$, can be interpreted as the maximal bias-strength. 
In the sequel, we call the decision making model composed of \req{logit} and \req{bias1} as Model~1. 
The block diagram of Model~1 is illustrated in \refig{logit_b1}.

\begin{figure}[t]
\begin{center}
\includegraphics[width=84mm]{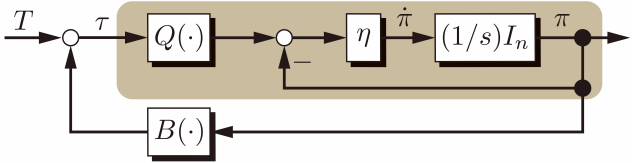}    
\caption{Block diagram of Model~1. 
The system colored by khaki is the logit model \req{logit}.
} 
\label{fig:logit_b1}
\end{center}
\end{figure}

In this paper, we address another model of conformity bias, presented in~\cite{CL2018, YHet2020}. 
For given $T$, in this model, the biased cost $\tau$ obeys 
\begin{align}
\tau = W(\pi) T, 
\label{eq:bias2}
\end{align}
where $W(\pi) = \textrm{diag} (w_1(\pi_1), w_2(\pi_2), \ldots, w_n(\pi_n))$ is the bias matrix. 
The $k$-th diagonal element $w_k: [0,1] \to \bR$ ($k \in \cA$) is the bias function for strategy $k$. 
Throughout this paper, we set the following assumption for the bias function $w_k$. 
\begin{assumption}
\label{as:func_w}
The function $w_k \ (k \in \cA)$ obeys the following items: 
\begin{itemize}
\item[(i)] $w_k$ is continuous in $[0, 1]$ and continuously differentiable in $(0, 1)$. 
\item[(ii)] 
$w_k$ is decreasing in $\pi_k$, and uniformly bounded by $[w^L, w^H]$ with some $w^H \geq w^L > 0$. 
\item[(iii)] 
The first derivative $\nabla w_k$ is uniformly bounded by $[- v^H, - v^L]$ with some $v^H \geq v^L > 0$. 
\end{itemize}
\end{assumption}

In the subsequent discussion, we call the model composed of \req{logit} and \req{bias2} as Model~2. 
Under \reas{func_w}, the biased cost for strategy $k$ decreases when $\pi_k$ gets large. 
In other words, the majority strategy tends to be well chosen. 
Since $| \nabla w_k |$ is the dependency of $w_k$ on $\pi_k$, this implies the strength of the bias and its maximum value $v^H$ represents the maximal bias-strength of Model~2. 
The block diagram of Model~2 is illustrated in \refig{logit_b2}.

\begin{figure}[t]
\begin{center}
\includegraphics[width=84mm]{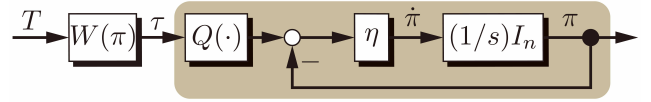}    
\caption{Block diagram of Model~2. 
The khaki part implies \req{logit}.
} 
\label{fig:logit_b2}
\end{center}
\end{figure}


\section{Passivity Analysis of Decision Making Models}

In this section, we analyze the biased decision making models, introduced in the last section, in terms of $\delta$-passivity.

First, we consider Model~1, which is composed of \req{logit} and \req{bias1}. 
From \req{bias1}, the signal $\dot{\tau}$ is given by 
\begin{align}
\dot{\tau} = \dot{T} + B'(\pi) \dot{\pi}, 
\label{eq:dtau1}
\end{align}
where $B'(\pi) := \textrm{diag} (\nabla b_1(\pi_1), \nabla b_2(\pi_2), \dots, \nabla b_n(\pi_n) )$ is negative definite because of \reas{func_b}. 
From \relem{dp0} and \req{dtau1}, we have the following lemma. 
\begin{lemma}
\label{lem:dps_bias1}
Under \reas{func_b}, the system \req{logit} with \req{bias1} is $\delta$-output-passivity-short from $- \dot{T}$ to $\dot{\pi}$.
\end{lemma}
\begin{proof}
Consider the function $S(\tau, \pi)$ defined in \req{stfunc1} with substituting $\tau = T + B (\pi)$. 
From \relem{dp0}, we have
\begin{align}
\dot{S} &\leq - \dot{\tau}^\top \dot{\pi} \nn\\
&= - (\dot{T} + B'(\pi) \dot{\pi})^\top \dot{\pi}
\nn\\
& = - \dot{T}^\top \dot{\pi} - \dot{\pi}^\top B'(\pi) \dot{\pi} 
\nn\\
&\leq  - \dot{T}^\top \dot{\pi} + c^H \| \dot{\pi}  \|^2. 
\label{eq:dps_bias1}
\end{align}
This completes the proof. 
\end{proof}

Focusing on the $\delta$-passivity, the block diagram of Model~1 can be represented as \refig{logit_b1_dp}. 
\relem{dps_bias1} suggests that a positive feedback by $-B'(\pi)$ appears on the outside of the non-biased logit model, as shown in \refig{logit_b1_dp}. 
We can confirm from \req{dps_bias1} that the impact coefficient incrementally varies according to strength of the bias $c^H$. 
In other words, the conformity bias modeled in \req{bias1} violates passivity of decision making.

\begin{figure*}[t]
\begin{center}
\includegraphics[width=130mm]{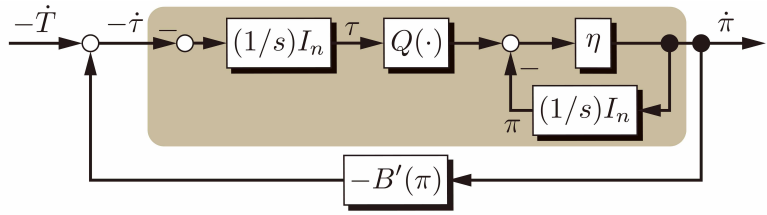}    
\caption{Block diagram of Model~1 expressed along $\delta$-passivity. 
The khaki part is $\delta$-passive (\relem{dp0}).
} 
\label{fig:logit_b1_dp}
\end{center}
\end{figure*}

\begin{figure*}[t]
\begin{center}
\includegraphics[width=130mm]{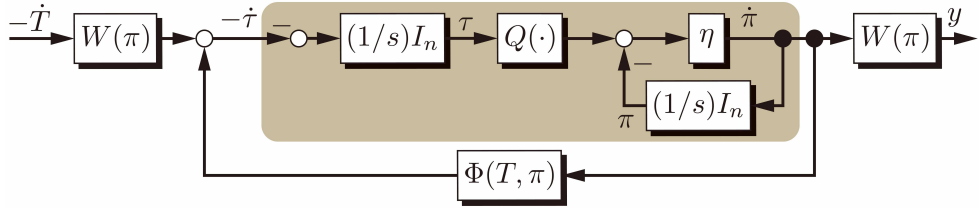}    
\caption{Block diagram of Model~2 expressed along $\delta$-passivity. 
} 
\label{fig:logit_b2_dp}
\end{center}
\end{figure*}

We next analyze Model~2. 
Here, we define the diagonal matrix $W'(\pi) := \textrm{diag} (\nabla w_1(\pi_1), \nabla w_2(\pi_2), \dots, \nabla w_n(\pi_n) )$ which is negative definite due to \reas{func_w}. 
Then, the bias model \req{bias2} yields 
\begin{align}
\dot{\tau} &= W(\pi) \dot{T} + \textrm{diag}(W'(\pi) \dot{\pi}) T
\nn\\
&= W(\pi) \dot{T} - \Phi (T, \pi) \dot{\pi}, 
\label{eq:dtau2}
\end{align}
where $\Phi (T, \pi) := - W'(\pi) \textrm{diag}(T)$ is a diagonal matrix. 
Substituting \req{dtau2} to the inequality in \relem{dp0}, we can prove the following lemma which is a particular case of \cite[Lemma~5]{YHet2020}. 
\begin{lemma}[\cite{YHet2020}]
\label{lem:dps_bias2}
Suppose that there exists a constant $T_{\max} > 0$ such that $T_k (t) \leq T_{\max}$ holds for all $k \in \cA$ and for all $t \geq 0$. 
Under \reas{func_w}, the system \req{logit} with \req{bias2} is $\delta$-output-passivity-short from $-\dot{T}$ to $y := W(\pi) \dot{\pi}$. 
\end{lemma}
\begin{proof}
Consider the function $S(\tau, \pi)$ with substituting $\tau = W(\pi)T$. 
From \relem{dp0} and \req{dtau2}, we obtain
\begin{align}
\dot{S} &\leq - \dot{\tau}^\top \dot{\pi} 
\nn\\
&= - \left( W(\pi) \dot{T} - \Phi (T, \pi) \dot{\pi} \right)^\top \dot{\pi}
\nn\\
& = - \dot{T}^\top W(\pi) \dot{\pi} + \dot{\pi}^\top \Phi (T, \pi) \dot{\pi}
\nn\\
&= - \dot{T}^\top y + y^\top (W(\pi))^{-1} \Phi (T, \pi) (W(\pi))^{-1} y
\nn\\
& \leq - \dot{T}^\top y + \frac{v^H T_{\max}}{(w^L)^2} \| y \|^2. 
\label{eq:dps_bias2}
\end{align}
This completes the proof. 
\end{proof}

The block diagram of Model~2 can be illustrated in \refig{logit_b2_dp} by replacing the input and the output as $- \dot{T}$ and $y$. 
\relem{dps_bias2} reveals that the matrix $\Phi (T, \pi)$ appears as a positive feedback on the original logit model \req{logit}, which is $\delta$-passive. 
The impact coefficient in \req{dps_bias2} increases with $v^H$ which is the strength of the bias $W(\pi)$. 
This result suggests that the conformity bias represented in \req{bias2} destabilizes the decision making.

Both of Lemmas~\ref{lem:dps_bias1} and \ref{lem:dps_bias2} indicate that the positive feedback paths attributed by the conformity biases
violate $\delta$-passivity of the original rational model \req{logit}. 
Therefore, we conclude that the conformity bias works to destabilize the dynamic decision making.

\section{Passivity-based Design of Mechanisms}

In this section, we propose behavior modification mechanisms to lead decision makers to desired social state. 
Particularly, we focus on the output-passivity-shortage of the decision making models, shown in the last section, and design mechanism based on passivity paradigm. 
In the sequel, we assume that the population state $\pi$ is observable, and the desired population state, denoted as $\pi^* \in \textrm{int}(\cS)$, is given as a constant vector.  
In this paper, a mechanism indicates the system to update the actual cost\footnote{In the case of incentive design, the cost $T(t)$ is added to decision makers as an economic input. Whereas, in the case of nudge, it is announced to them as an informational input.} $T$ by using $\pi$ and $\pi^*$. 
The goal in this section is to design mechanisms so as to ensure the inducement $\lim_{t \to \infty} \pi(t) = \pi^*$ for Model~1 and Model~2, respectively. 

\subsection{Mechanism for Model~1}
\label{ssec:nudge1}

In this subsection, we consider a mechanism for Model~1. 
The structure of the mechanism is illustrated in \refig{nudge1}. 
The gray part in \refig{nudge1} implies Model~1, and the block $\Sigma_1$ is the mechanism. 
Notice that the system enclosed by the red line is $\delta$-output-passivity-short, as proved in \relem{dps_bias1}. 
In terms of passivity theory, positive energy generated from an output-passivity-short system can be canceled out by negative feedback of an input-strictly-passive system. 
Hence, we can expect to implement a stable mechanism by designing $\Sigma_1$ so as to satisfy input-strict-passivity of the system enclosed by blue line. 

\begin{figure*}[t]
\begin{center}
\includegraphics[width=160mm]
{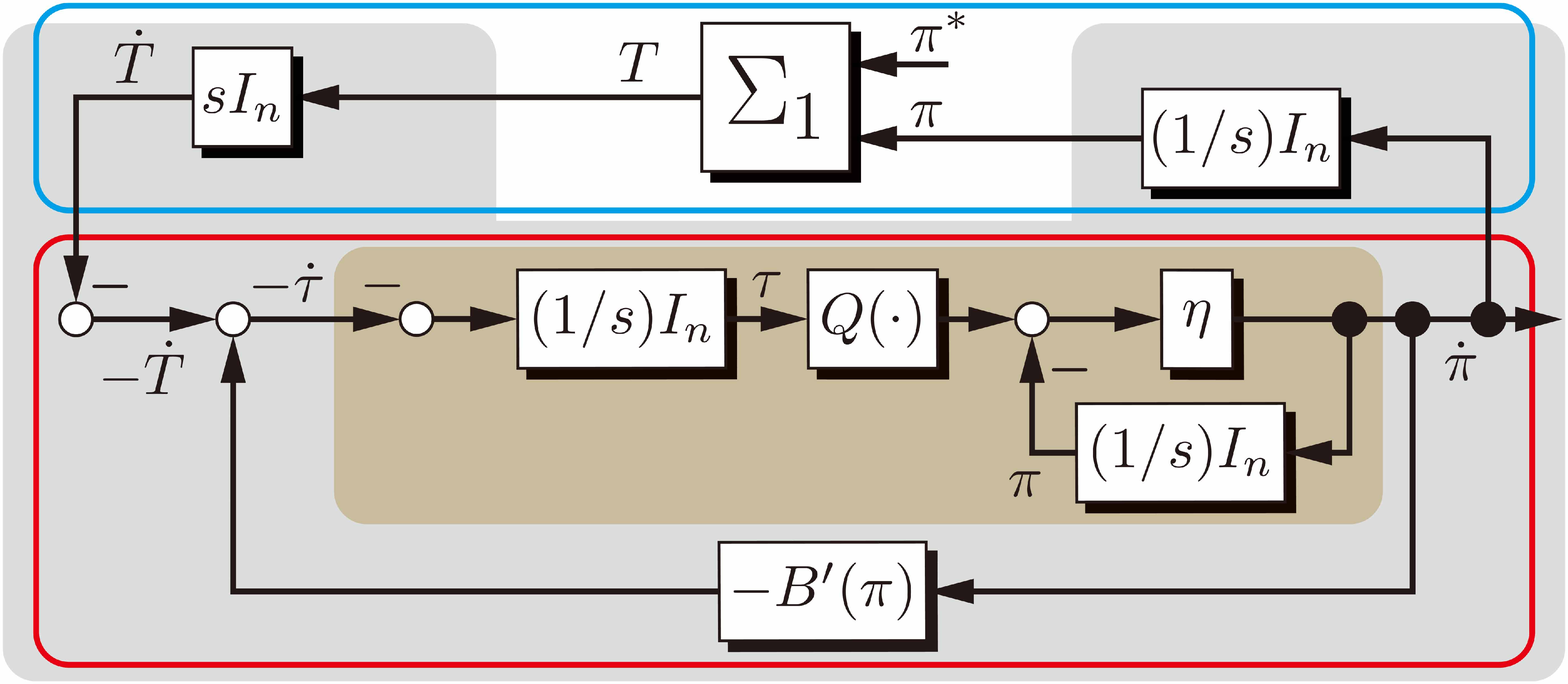}
\caption{Block diagram of a mechanism for Model~1. 
The gray part corresponds to the decision makers. 
The system enclosed by the red line is $\delta$-output-passivity-short (\relem{dps_bias1}). 
The system $\Sigma_1$ is a mechanism for Model~1 to make the system of blue line $\delta$-passive. 
} 
\label{fig:nudge1}
\end{center}
\end{figure*}

Based on the above concept, we propose the mechanism inspired by Proportional-Integral controller, as below: 
\begin{subequations}
\begin{align}
\dot{\mu} &= \rho \left( \pi - \pi^* \right), 
\label{eq:nudge1_1}
\\
T &= \mu + \kappa (\pi - \pi^*), 
\label{eq:nudge1_2}
\end{align}\label{eq:nudge1}\end{subequations}
where $\rho > 0$ and $\kappa > 0$ is a design parameter. 
The mechanism \req{nudge1} satisfies the following lemma. 
\begin{lemma}
\label{lem:dip_n1}
The system \req{nudge1} is $\delta$-input-strictly-passive from $\dot{\pi}$ to $\dot{T}$ for the storage function $H(\pi) := \frac{\rho}{2}\| \pi - \pi^* \|^2$. 
\end{lemma}
\begin{proof}
From \req{nudge1_2}, $\dot{T} = \dot{\mu} + \kappa \dot{\pi}$ holds. 
Noticing this result and \req{nudge1_1}, we have
\begin{align}
\dot{H} &= \rho (\pi - \pi^*)^\top \dot{\pi} 
= \dot{\mu}^\top \dot{\pi} 
\nn\\
&= \dot{T}^\top \dot{\pi} - \kappa \| \dot{\pi} \|^2 . 
\label{eq:dip_n1}
\end{align}
This result shows the $\delta$-input-strict-passivity of \req{nudge1}. 
\end{proof}

The result in \relem{dip_n1} suggests that the system from $\dot{\pi}$ to $\dot{T}$, enclosed by the blue line in \refig{nudge1}, is $\delta$-input-strictly-passive if we apply \req{nudge1} to $\Sigma_1$. 
Thus, it is expected that the mechanism \req{nudge1} cancels out the positive energy in \req{dps_bias1} and induces the population state $\pi$ to the desired state $\pi^*$. 

As a preparation for convergence analysis, we introduce the following lemma about the storage function $S$ under the mechanism \req{nudge1}. 
\begin{lemma}
\label{lem:Tinf1}
Consider the system \req{nudge1}. 
Let a signal $\pi(t) \in \textrm{int}(\cS)$ achieves $\| T (t) \| \to \infty$ when $t \to \infty$. 
For this specific signal $\pi(t)$, $S(T + B(\pi), \pi) \to \infty$ holds for any $\pi^* \in \textrm{int}(\cS)$. 
\end{lemma}
\begin{proof}
See \reapp{pr_Tinf1}. 
\end{proof}

We are now ready to prove the following theorem. 
In the proof, we use the notation $\cL_\infty$ as the set of all signals $u: [0,\infty) \to \bR^n$ satisfying $\sup_{t \in [0,\infty)} \| u(t) \| < \infty$. 
\begin{theorem}
\label{th:conv1}
Consider the feedback system composed of \req{logit}, \req{bias1} and \req{nudge1}. 
When $\kappa > c^H$ holds, the population state $\pi$ achieves $\lim_{t \to \infty} \pi(t) = \pi^*$. 
\end{theorem}
\begin{proof}
We define the function $V_1 := S(T + B(\pi), \pi) + H(\pi)$. 
From \relem{dps_bias1} and \relem{dip_n1}, we obtain
\begin{align}
\dot{V}_1 &\leq - \dot{T}^\top \dot{\pi} + c^H \| \dot{\pi} \|^2 + \dot{T}^\top \dot{\pi} - \kappa \| \dot{\pi} \|^2
\nn\\
&= - (\kappa - c^H) \| \dot{\pi} \|^2 \leq 0
\label{eq:energy1}
\end{align}
under $\kappa > c^H$. 
Denote the initial values of $S$ and $H$ by $S_0$ and $H_0$, respectively. 
Then, $V_1 \leq S_0 + H_0 < \infty$ holds for any time $t \geq 0$. 
Remarking $\pi \in \textrm{int}(\cS)$, this implies $T \in \cL_\infty$ due to \relem{Tinf1}. 
Thus, we can apply the LaSalle's invariance principle~\cite{K2002}, and hence solution of \req{logit}, \req{bias1} and \req{nudge1} for any initial conditions $T(0) \in \bR^n$ and $\pi(0) \in \textrm{int}(\cS)$ converges to the largest invariant set satisfying $\dot{V}_1 = 0$. 

Consider the state trajectories such that $\dot{V}_1 \equiv 0$ holds. 
From \req{energy1}, $\dot{V}_1 \equiv 0$ implies $\dot{\pi} \equiv 0$. 
Hence, $\pi$ should be constant. 
Focusing on \req{nudge1}, $\dot{T} \equiv \dot{\mu}$ holds and $\dot{\mu} = \rho(\pi - \pi^*)$ should be constant. 
If $\pi \neq \pi^*$, $\dot{T}\neq 0$ identically holds and thus $T$ should diverge. 
However, the divergence contradicts $T \in \cL_\infty$. 
Accordingly, $T$ is constant and $\pi \equiv \pi^*$ is satisfied. 

As a result, by invoking the LaSalle's invariance principle, we can prove that $\pi$ asymptotically converges to $\pi^*$. 
\end{proof}

\reth{conv1} shows that the mechanism \req{nudge1} can induce the decision makers of Model~1 to desired state if we design the parameter $\kappa$ large enough. 
The right hand side of the convergence condition $\kappa > c^H$ is given by the bias-strength. 
Hence, for the strongly biased decision makers, $\kappa$ is needed to be set large.

\subsection{Mechanism for Model~2}

Let us consider to design a mechanism for Model~2. 
Its structure is illustrated in \refig{nudge2}: The gray part is the decision makers, and $\Sigma_2$ is the mechanism. 
Similar to \ressec{nudge1}, we design $\Sigma_2$ to guarantee the passivity from $y$ to $\dot{T}$. 
\begin{figure*}[t]
\begin{center}
\includegraphics[width=160mm]
{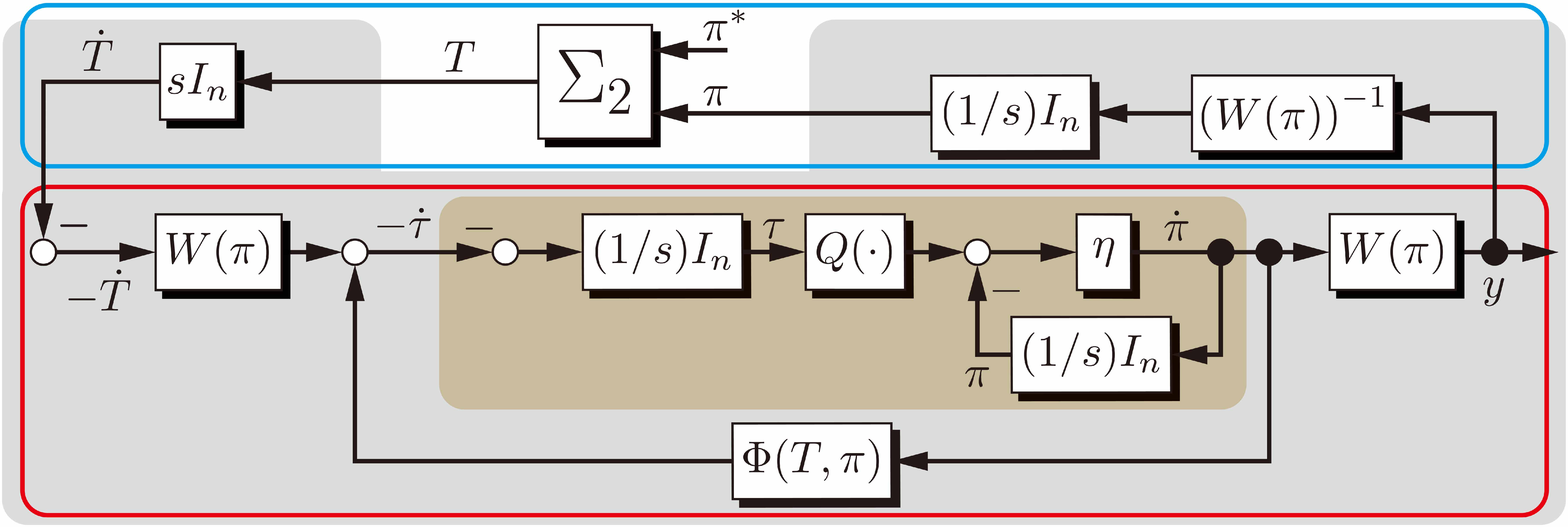}
\caption{Block diagram of a mechanism for Model~2. 
The gray part corresponds to the decision makers. 
The system enclosed by the red line is $\delta$-output-passivity-short (\relem{dps_bias2}). 
The system $\Sigma_2$ is a mechanism for Model~2 to make the system in blue line $\delta$-passive. 
} 
\label{fig:nudge2}
\end{center}
\end{figure*}

Applying the mechanism \req{nudge1} to $\Sigma_2$, the following $\delta$-passivity is satisfied, that is proved in the previous work~\cite{YHet2020}. 
\begin{fact}[\cite{YHet2020}]
Under \reas{func_w}, the system \req{nudge1} is $\delta$-input-strictly-passive from $y$ to $\dot{T}$.
\end{fact}
By using \req{nudge1} for $\Sigma_2$, certainly, the whole system in \refig{nudge2} is composed of the feedback of a $\delta$-output-passivity-short system and a $\delta$-input-strictly-passive system. 
However, the analysis in~\cite{YHet2020} leaves an important issue about signal boundedness. 
In~\cite{YHet2020}, the boundedness of the signal $T$ is given by an assumption, which is not theoretically guaranteed. 
In particular, to certify the upper bound of signal $T$, denoted by $T_{\max}$, is the most significant problem to cancel out the positive energy in \req{dps_bias2}. 
Therefore, we should redesign a mechanism for Model~2 to guarantee the existence of $T_{\max}$, with keeping passivity.

To solve the above issue, we propose the following mechanism. 
\begin{subequations}
\begin{align}
\dot{\mu} &= \min \left \{ \rho \left( \pi - \pi^* \right), - \alpha \mu \right \} , \ \mu(0) \leq 0,  
\label{eq:nudge2_1}
\\
T &= \bar{T} \bone_n + \mu + \kappa (\pi - \pi^*), 
\label{eq:nudge2_2}
\end{align}\label{eq:nudge2}\end{subequations}
where $\rho > 0$ and $\kappa > 0$ are design parameters, $\alpha > 0$ is a parameter satisfying $\alpha > \frac{1}{2 \kappa}$, and $\bar{T} > 0$ is a constant. 
Remark that \req{nudge2_1} ensures $\mu \leq 0$ for all time. 
Thus, the system \req{nudge2} guarantees the existence of $T_{\max} < \bar{T} + \kappa$ such that $T_k \leq T_{\max} \ \forall k \in \cA$ holds for all time. 
About the time derivative of $T$, $\dot{T} = \dot{\mu} + \kappa \dot{\pi}$ holds from \req{nudge2_2}.

For the passivity analysis of \req{nudge2}, we consider the primitive function of $w_k$. 
Under \reas{func_w}, $w_k$ is continuous. 
Hence, we can define the following function: 
\begin{align}
F_k (\pi_k) &:= 
\left\{
\begin{array}{cl}
 \displaystyle{\int_0^{\pi_k} \left( \int_0^{\phi} w_k (\theta) d \theta \right)  d\phi} & \textrm{if} \ \pi_k \in [0, 1],
\\
\infty & \textrm{otherwise}.
\end{array}
\right.
\end{align}
For $\pi_k \in (0, 1)$, $\nabla^2 F_k (\pi_k) = w_k (\pi_k) \geq w^L$ clearly holds and this reveals that $F_k$ is convex. 
We denote $F^\star_k : \bR \to \bR$ as the convex conjugate of $F_k$ \cite{BV2004}. 
Due to the convexity of $F_k$, 
\begin{align}
\pi_k = \nabla F^\star_k (\nabla F_k(\pi_k)), \ \pi_k \in (0, 1) 
\end{align}
is satisfied~\cite{BV2004}. 
By using these results, we can prove the following lemma. 
\begin{lemma}
\label{lem:dip_n2}
Define $\zeta_k := \nabla F_k(\min \{ \pi_k, \pi_k^* - \frac{\alpha}{\rho} \mu_k \})$ and $\zeta_k^* := \nabla F_k (\pi_k^*)$ for each $k \in \cA$. 
Under \reas{func_w}, the system \req{nudge2} is $\delta$-input-strictly-passive from $y$ to $\dot{T}$
for the storage function
\begin{align}
U(\mu, \pi) &:= \sum_{k \in \cA} U_k (\mu_k, \pi_k), 
\nn\\
U_k (\mu_k, \pi_k) &:= \rho \left( F^\star_k (\zeta_k) - F^\star_k (\zeta_k^*) - (\zeta_k - \zeta_k^*) \nabla F^\star_k (\zeta_k^*)  \right).
\nn
\end{align}
\end{lemma}
\begin{proof}
Due to the convexity of $F^\star_k$, 
\begin{align*}
	F^\star_k (\zeta_k) - F^\star_k (\zeta_k^*) \geq (\zeta_k - \zeta_k^*) \nabla F^\star_k (\zeta_k^*)
\end{align*}
holds 
and hence the function $U_k$ becomes positive semi-definite. 
We first consider the time when mode switch does not happen in \req{nudge2_1} for $k$-th element. 
Then, the time derivative of $U_k$ along \req{nudge2} is given by
\begin{align}
\dot{U}_k = \rho \left(\nabla F^\star_k (\zeta_k) -  \nabla F^\star_k (\zeta_k^*) \right) \dot{\zeta}_k.
\end{align}
If $\pi_k - \pi_k^* \leq - \frac{\alpha}{\rho} \mu_k$, then $\zeta_k = \nabla F_k(\pi_k)$ holds and hence
\begin{align}
\dot{U}_k &= \rho (\pi_k - \pi_k^*)\nabla^2 F_k (\pi_k) \dot{\pi}_k
\nn\\
&= \rho (\pi_k - \pi_k^*) y_k.
\end{align}
Under $\pi_k - \pi_k^* \leq - \frac{\alpha}{\rho} \mu_k$, $\dot{T}_k = \rho \left( \pi_k - \pi_k^* \right) + \kappa \dot{\pi}_k$ holds. 
Thus, we have
\begin{align}
\dot{U}_k &= \dot{T}_k y_k - \kappa \dot{\pi}_k y_k
\nn\\
&=  \dot{T}_k y_k - \kappa w_k(\pi_k) \dot{\pi}_k^2 
\nn\\
&\leq  \dot{T}_k y_k - \frac{\kappa}{2} w_k(\pi_k) \dot{\pi}_k^2. 
\label{eq:ineq_mode1}
\end{align}
If $\pi_k - \pi_k^* > - \frac{\alpha}{\rho} \mu_k$, we obtain $\zeta_k = \nabla F_k(\pi_k^* - \frac{\alpha}{\rho} \mu_k)$ and $\dot{\mu}_k = - \alpha \mu_k$. 
Then, 
\begin{align}
\dot{U}_k &= \rho \left( \pi_k^* - \frac{\alpha}{\rho} \mu_k - \pi_k^* \right) \nabla^2 F_k \left( \pi_k^* -  \frac{\alpha}{\rho}  \mu_k \right) (- \alpha \dot{\mu}_k)
\nn\\
&= \alpha^2 w_k \left( \pi_k^* -  \frac{\alpha}{\rho}  \mu_k \right) \mu_k \dot{\mu}_k 
\nn\\
&= - \alpha w_k \left( \pi_k^* -  \frac{\alpha}{\rho}  \mu_k \right) \dot{\mu}_k^2
\end{align}
holds. 
Since $w_k$ is decreasing function, we obtain $w_k (\pi_k^* - \frac{\alpha}{\rho} \mu_k) \geq w_k (\pi_k)$. 
From this inequality and $\dot{T} = \dot{\mu} + \kappa \dot{\pi}$, we have
\begin{align}
\dot{U}_k &\leq - \alpha w_k (\pi_k)  \dot{\mu}_k^2 
\nn\\
&= - \frac{1}{2 \kappa} w_k (\pi_k)  \dot{\mu}_k^2 - \left( \alpha - \frac{1}{2 \kappa} \right)w_k (\pi_k)  \dot{\mu}_k^2
\nn\\
&\leq - \frac{1}{2 \kappa} w_k (\pi_k)  \left( \dot{T}_k - \kappa \dot{\pi}_k \right)^2 
\nn\\
&= \dot{T}_k w_k (\pi_k)  \dot{\pi}_k -\frac{1}{2 \kappa}  w_k (\pi_k) \dot{T}_k^2 - \frac{\kappa}{2} w_k (\pi_k)  \dot{\pi}_k^2 
\nn\\
& \leq \dot{T}_k y_k - \frac{\kappa}{2} w_k (\pi_k)  \dot{\pi}_k^2. 
\label{eq:ineq_mode2}
\end{align}

Next, we consider the time when a mode switch occurs in \req{nudge2_1} for $k$-th element. 
Since $U_k (\mu_k,\pi_k)$ is not differentiable at $(\mu_k,\pi_k)$ satisfying $\pi_k - \pi_k^* = -\alpha \mu_k$, we now take in the upper Dini derivative, denoted by $D^+ U_k$. 
From the results in \req{ineq_mode1} and \req{ineq_mode2}, we can confirm that
\begin{align*}
D^+ U_k \leq \dot{T}_k y_k - \frac{\kappa}{2} w_k (\pi_k)  \dot{\pi}_k^2 
\end{align*}
holds for all time $t \geq 0$. 
Therefore, we obtain
\begin{align}
D^+ U &\leq \dot{T}^\top y - \frac{\kappa}{2} \dot{\pi}^\top W(\pi)  \dot{\pi}  
\nn\\
&=  \dot{T}^\top y - \frac{\kappa}{2}y^\top \left( W(\pi) \right)^{-1} y 
\nn\\
& \leq \dot{T}^\top y - \frac{\kappa}{2 w^H} \| y \|^2. 
\label{eq:diniU}
\end{align}
Integrating \req{diniU} in time completes the proof. 
\end{proof}

Thanks to the result in \relem{dip_n2}, it is revealed that the proposed mechanism \req{nudge2} guarantees both $\delta$-input-strict-passivity and the existence of $T_{\max}$. 
In other words, the remained issue in the previous method~\cite{YHet2020} can be cleared by \req{nudge2}.
Thus, we can expected to achieve the passivity-based mechanism for Model~2.

Here we introduce the following lemma, which will be used in convergence analysis. 
\begin{lemma}
\label{lem:Tinf2}
Consider the system \req{nudge2}. 
Let a signal $\pi(t) \in \textrm{int}(\cS)$ achieves $\| T (t) \| \to \infty$ when $t \to \infty$. 
For this specific signal $\pi(t)$, $S(W(\pi) T, \pi) \to \infty$ holds for any $\pi^* \in \textrm{int}(\cS)$. 
\end{lemma}
\begin{proof}
See \reapp{pr_Tinf2}. 
\end{proof}


From \relem{dps_bias2}, \relem{dip_n2} and \relem{Tinf2}, we are now ready to analyze convergence of the proposed mechanism \req{nudge2} for Model~2. 
We should remark that the LaSalle's invariance principle~\cite{K2002} cannot be applied due to non-smoothness of the storage function $U(\mu,\pi)$. 
Alternatively, we discuss the convergence by using the invariance principle for non-smooth Lyapunov functions~\cite{SP1994}. 
Then, the following theorem can be proven. 
\begin{theorem}
\label{th:conv2}
Consider the feedback system composed of \req{logit}, \req{bias2} and \req{nudge2} under \reas{func_w}. 
When $\kappa > \frac{2 v^H T_{\max}}{w^L}$ holds, the population state $\pi$ achieves $\lim_{t \to \infty} \pi(t) = \pi^*$. 
\end{theorem}
\begin{proof}
Define the function $V_2 := S(W(\pi)T, \pi) + U(\mu, \pi)$. 
From  \relem{dps_bias2} and \relem{dip_n2}, the upper Dini derivative of $V_2$ satisfies
\begin{align}
D^+ V_2 &\leq - \dot{T}^\top y + \dot{\pi}^\top \Phi (T, \pi) \dot{\pi}
 + \dot{T}^\top y - \frac{\kappa}{2} \dot{\pi}^\top W(\pi)  \dot{\pi}
\nn\\
&= \dot{\pi}^\top \left( \Phi (T, \pi)  - \frac{\kappa}{2} W(\pi)  \right) \dot{\pi}. 
\label{eq:energy2}
\end{align}
Under $\kappa > \frac{2 v^H T_{\max}}{w^L}$, $\Phi (T, \pi)  - \frac{\kappa}{2} W(\pi) \prec 0$ is satisfied and hence $D^+ V_2 \leq 0$ holds. 
Denote the initial values of $S$ and $U$ by $S_0$ and $U_0$, respectively. 
Then, $V_2 \leq S_0 + U_0 < \infty$ holds for all time. 
This implies $T \in \cL_\infty$ and $\mu \in \cL_\infty$ from \relem{Tinf2} and $\pi \in \textrm{int}(\cS)$. 
Accordingly, the invariance principle for non-smooth function~\cite{SP1994} is applicable, 
and trajectories generated by \req{logit}, \req{bias2} and \req{nudge2} for any initial conditions $T(0) \in \bR^n$ and $\pi(0) \in \textrm{int}(\cS)$ converge to the largest invariant set satisfying $D^+ {V}_2 = 0$. 

Let us now suppose the situation under $D^+ V_2 \equiv 0$. 
From \req{energy2}, $D^+ V_2 \equiv 0$ yields $\dot{\pi} \equiv 0$ and hence $\pi$ must be constant. 
If $\pi_k < \pi^*_k$, then $k$-th element of \req{nudge2_1} follows $\dot{\mu}_k \equiv \rho(\pi_k - \pi_k^*) < 0$, which contradicts the boundedness of $\mu$. 
Thus, $\pi \geq \pi^*$ is identically satisfied. 
If $\pi_k > \pi^*_k$ holds for some $k \in \cA$, $\sum_{l \in \cA} (\pi_l - \pi_l^*) \geq \pi_k - \pi_k^* > 0$ must hold. 
However, this contradicts the fact $\bone_n^\top (\pi - \pi^*) \equiv 0$ given by $\pi \in \textrm{int}(\cS)$ and $\pi^* \in \textrm{int}(\cS)$. 
Therefore, $\pi \equiv \pi^*$ is satisfied when $D^+ V_2 \equiv 0$ holds. 

In summary, the invariance principle for non-smooth function~\cite{SP1994} proves that $\pi$ asymptotically converges to $\pi^*$. 
\end{proof}

Thanks to the upper bound condition $T_k(t) \leq T_{\max} \ \forall k\in \cA, \forall t \geq 0$ ensured by \req{nudge2}, the stability and convergence of the feedback system in \refig{nudge2} can be exactly guaranteed under the gain condition  $\kappa > \frac{2 v^H T_{\max}}{w^L}$. 
Similar to \reth{conv1}, the result in \reth{conv2} suggests the tendency that large gain $\kappa$ will be required for strongly biased decision makers. 
The quantitative inequality is calculated as a result of passivity-based analysis. 

\begin{remark}
	\label{rem:CL2018_1}
	The authors in \cite{CL2016} and \cite{CL2018} proposed a similar nudging mechanism to \req{nudge1}, with assuming that the signal $Q(\tau)$ is observable. 
	This assumption is different from the one dealt with in this paper. 
	Whereas, the information $Q(\tau)$ is not easy to get from the decision makers since it is implicit variable. 
	Although $Q(\tau)$ might be estimated by using $\pi$ and $\dot{\pi}$, there is another difficulty of identifying $\eta$. 
	Thus, we suppose the observation of population state $\pi$, which is commonly used in the field of evolutionary game \cite{PSM2018}. 
\end{remark}

\begin{remark}
	In \cite{CL2018}, the authors addressed the biased population dynamics, and analyzed the convergence of nudging mechanism based on singular perturbation. 
	The convergence condition in \cite{CL2018} implicitly relies on the update rate $\eta$ of the decision makers \req{logit}. 
	Meanwhile, the result in \reth{conv2} shows two advantages against \cite{CL2018}. 
	The first one is to explicitly clarify a quantitative condition for stability, which is a benefit of passivity-based analysis. 
	In addition, the proposed mechanism in this paper can design the parameter $\kappa$ independently of the update rate $\eta$, which is the second contribution against the nudging method in \cite{CL2018}. 
\end{remark}


\section{Conclusion}

In this paper, we addressed design of mechanisms for decision makers with conformity biases. 
We first introduced two types of bias models addressed in~\cite{BD2001} and~\cite{CL2018}. 
Next, we analyzed the population dynamics with the bias models in terms of $\delta$-passivity. 
Then, we clarified that conformity biases appear as positive feedback terms, and they break passivity of dynamic decision making. 
We furthermore presented passivity-based mechanisms for biased population dynamics, and showed convergence conditions for the proposed mechanisms. 
Accordingly, we confirmed that high gain feedback should have been required for the decision makers with strong biases. 


\appendices

\section{Radial Unboundedness of Storage Function}
\label{app:pr_Srad}

The storage function $S(\tau, \pi)$, defined in \req{stfunc1}, satisfies the following lemma. 
\begin{lemma}
\label{lem:Srad}
Define 
$
\tau^H  := \max\{ \tau_1, \tau_2, \dots, \tau_n \}$ and $
\tau^L  := \min\{ \tau_1, \tau_2, \dots, \tau_n \}$. 
Under $\pi \in \rm{int}(\cS)$, then 
\begin{align*}
\tau^H - \tau^L \to \infty \ \Rightarrow \ S(\tau, \pi) \to \infty
\end{align*}
holds. 
\end{lemma}
\begin{proof}
Let us first focus on the second term in the right hand side of \req{stfunc1}, which is a convex optimization about $\omega$. 
Remarking that $\sum_{l \in \cA}\omega_l \log{\omega_l}$ works as a barrier function to the inequality constraints, $\omega^* \in \bR^n$ is an optimal solution if and only if there exists $\lambda^* \in \bR$ satisfying the following conditions: 
\begin{subequations}
\begin{align}
& \tau_k + \frac{1}{\beta} \left( 1 +  \log \omega_k^* \right) + \lambda^* = 0 \ \ \forall k \in \cA, 
\label{eq:kkt1}
\\
&  {\bf 1}_n^\top \omega^* = 1.
\label{eq:kkt2}
\end{align}\end{subequations}
From \req{kkt1}, we obtain $ \log \omega_k^* = - \beta \tau_k -  \beta \lambda^*  - 1$ and hence 
\begin{align}
\omega_k^* = \textrm{exp} \left( - \beta \tau_k -  \beta \lambda^*  - 1 \right). 
\label{eq:kkt1_2}
\end{align}
Applying \req{kkt2} to \req{kkt1_2}, we have
\begin{align}
& \textrm{exp} \left( -  \beta \lambda^*  - 1 \right) \sum_{l \in \cA} \textrm{exp} \left( - \beta \tau_l\right) = 1, 
\nn\\
\therefore \ & \lambda^* = \frac{1}{\beta} \left( - 1 + \log \left( \sum_{l \in \cA} \textrm{exp} \left( - \beta \tau_l\right) \right)  \right). 
\label{eq:kkt1_3}
\end{align}
Hence, we can calculate the second term in the right hand side of \req{stfunc1} as 
\begin{align*}
& \min_{\omega \in \cS} \left( \omega^\top \tau+\frac{1}{\beta}\sum_{l \in \cA}\omega_l \log{\omega_l} \right)
\\
& \ \ \ \ = (\omega^*)^\top \tau+\frac{1}{\beta}\sum_{l \in \cA}\omega^*_l \log{\omega^*_l}
\\
& \ \ \ \ = \sum_{l \in \cA}  \omega^*_l \left(\tau_l + \frac{1}{\beta} \left( - \beta \tau_l -  \beta \lambda^*  - 1 \right) \right)
\\
& \ \ \ \ = \sum_{l \in \cA}  \omega^*_l \left( - \lambda^*  - \frac{1}{\beta} \right) 
 = - \frac{1}{\beta}  \log \left( \sum_{l \in \cA} \textrm{exp} \left( - \beta \tau_l\right) \right). 
\end{align*}
Thus, the storage function $S(\tau, \pi)$ is given as 
\begin{align*}
S(\tau,\pi) = \ & \eta \pi^\top \tau+\frac{\eta }{\beta}\sum_{l \in \cA}\pi_l \log{\pi_l} 
+ \frac{\eta}{\beta} \log \left( \sum_{l \in \cA} \exp(-\beta \tau_l) \right). 
\end{align*}

Let us next focus on the terms depending on $\tau$. 
Due to the property of the Log-Sum-Exp function~\cite{BV2004}, 
\begin{align*}
\frac{1}{\beta} \log \left( \sum_{l \in \cA} \exp(-\beta \tau_l) \right) \geq - \tau^L
\end{align*}
is satisfied. 
Noticing that $\bone_n^\top \pi = 1$ holds and there exists $\veps > 0$ such that $\pi_k > \veps \ \forall k \in \cA$, we obtain 
\begin{align}
\pi^\top \tau + \frac{1}{\beta} \log \left( \sum_{l \in \cA} \exp(-\beta \tau_l) \right) 
&\geq \pi^\top \tau - \tau^L 
\nn\\
& = \sum_{l \in \cA} \pi_l (\tau_l - \tau^L)
\nn\\
&\geq \veps (\tau^H - \tau^L). 
\label{eq:app_ineq}
\end{align}

We now consider the case that $(\tau^H - \tau^L) \to \infty$ happens. 
Then, $\eta \pi^\top \tau + \frac{\eta}{\beta} \log \left( \sum_{l \in \cA} \exp(-\beta \tau_l) \right) $ goes to infinity due to \req{app_ineq}. 
Therefore, since $\sum_{l \in \cA}\pi_l \log{\pi_l}$ is finite, $S(\tau, \pi) \to \infty$ holds. 
This completes the proof. 
\end{proof}

\section{Proof of Lemmas}

\subsection{Proof of \relem{Tinf1}}
\label{app:pr_Tinf1}

Due to $\bone_n^\top \pi =1$ and $\bone_n^\top \pi^* = 1$, the signal $\dot{\mu}$ of \req{nudge1_1} satisfies $\bone_n^\top \dot{\mu} \equiv 0$. 
Hence, the system \req{nudge1} satisfies $\bone_n^\top T(t) \equiv \bone_n^\top \mu(t) \equiv \bone_n^\top \mu(0)$. 
Define 
\begin{align*}
T^H (t) &:= \max\{ T_1 (t), T_2(t), \dots, T_n(t) \}, 
\\
T^L (t) &:= \min\{ T_1 (t), T_2(t), \dots, T_n(t) \}. 
\end{align*}
When $\| T(t) \| \to \infty$ happens, $T^H(t) \to \infty$ and $T^L(t) \to - \infty$ hold due to the constraint $\bone_n^\top T(t) \equiv \bone_n^\top \mu(0)$. 
Then, $(T^H(t) - T^L(t) +b^L- b^H) \to \infty$ is also satisfied. 
Let us consider $\tau^H$ and $\tau^L$ defined in \relem{Srad} under \req{bias1}. 
Due to $\tau^H \geq T^H + b^L$ and $\tau^L \leq T^L + b^H$, we have 
$\tau^H - \tau^L \geq T^H - T^L + b^L - b^H$. 
Thus, $(\tau^H - \tau^L) \to \infty$ holds when $\| T(t) \| \to \infty$ happens. 
From the above discussion and \relem{Srad}, when the system \req{nudge1} achieves $\| T(t) \| \to \infty$ under $\pi(t) \in \textrm{int}(\cS)$, $S(T + B(\pi), \pi) \to \infty$ holds. 
This completes the proof of \relem{Tinf1}. 

\subsection{Proof of \relem{Tinf2}}
\label{app:pr_Tinf2}

Before we prove \relem{Tinf2}, we show the following lemma. 
\begin{lemma}
\label{lem:bound_n2}
Consider the signal $T(t)$ generated by \req{nudge2}. 
For any constant $\pi^* \in \textrm{int}(\cS)$ and for any signal $\pi(t) \in \textrm{int}(\cS)$, 
the signal $T^H (t) := \max\{ T_1 (t), T_2(t), \dots, T_n(t) \}$ is bounded. 
\end{lemma}
\begin{proof}
Let us suppose the case that the system \req{nudge2_1} satisfies $\| \mu_k \| \to \infty$ for all $k \in \cA$. 
Since \req{nudge2_1} guarantees $\mu_k \leq 0$, $\mu_k \to - \infty$ must hold for all $k \in \cA$. 
Then, there exists a time $\hat{t} > 0$ such that 
$- \frac{\alpha}{\rho} \mu_k > \pi_k - \pi_k^* \ \forall k \in \cA$
holds for any time $t \geq \hat{t}$. 
In other words, \req{nudge2_1} follows $\dot{\mu} = \rho(\pi - \pi^*)$ for any $t \geq \hat{t}$. 
This yields $\bone_n^\top \dot{\mu} = 0$ and hence $\sum_{l \in \cA} \mu_l (t)$ must be constant for any time $t \geq \hat{t}$. 
This contradicts the assumption satisfying $\mu_k \to - \infty \ \forall k \in \cA$. 

Therefore, the system \req{nudge2_1} guarantees at least $- \infty < \max \{ \mu_1(t), \mu_2(t), \dots, \mu_n (t) \}$. 
As a result, $T^H(t)$ is bounded since $\max \{ \mu_1(t), \mu_2(t), \dots, \mu_n (t) \} + \bar{T} -\kappa < T^H(t) < \bar{T} + \kappa$ holds under \req{nudge2_2}. 
\end{proof}

We now prove \relem{Tinf2}. 
Let us introduce the signals $T^H(t)$ and $T^L(t)$, defined in \reapp{pr_Tinf1}. 
If \req{nudge2} achieves $\|T(t)\| \to \infty$, then $T^L(t) \to - \infty$ must hold. 
This yields $(T^H(t) - T^L(t) ) \to \infty$ from \relem{bound_n2}. 
Then, $(w^H T^H(t) - w^L T^L(t) )\to \infty$ is also satisfied. 
Let us consider $\tau^H$ and $\tau^L$ defined in \relem{Srad} under \req{bias2}. 
Due to $\tau^H - \tau^L \geq w^H T^H - w^L T^L$, $(\tau^H - \tau^L) \to \infty$ holds when $\|T(t)\| \to \infty$ happens. 
Accordingly, \relem{Srad} proves \relem{Tinf2}.

\end{document}